\documentclass[12pt]{article}

\usepackage{amsfonts}
\usepackage{amssymb}
\usepackage{amsmath}
\usepackage{amsthm}
\usepackage{verbatim}
\usepackage{color}
\usepackage{algorithm}
\usepackage{algorithmic}
\usepackage{graphicx}
\usepackage{float}

\usepackage[top=1in, bottom=1in, left=1in, right=1in]{geometry}

\newtheorem{theorem}{Theorem}

\newtheorem{rem1}[theorem]{Remark}

\newtheorem{definition}[theorem]{Definition}
\newtheorem{proposition}[theorem]{Proposition}

\newtheorem{example}[theorem]{Example}
\newtheorem{ass1}[theorem]{Assumption}
\newtheorem{alg1}[theorem]{Algorithm}

\newenvironment{remark}{\begin{rem1}\rm}{\end{rem1}}
\newenvironment{assumption}{\begin{ass1}\rm}{\end{ass1}}

\numberwithin{equation}{section}
\numberwithin{theorem}{section}

\newcommand{\N}{\mathbb{N}}
\renewcommand{\P}{\mathbb{P}}

\newcommand{\R}{\mathbb{R}}

\newcommand{\T}{\mathbb{T}}

\newcommand{\lrparen}[1]{\left(#1\right)}

\newcommand{\lrsquare}[1]{\left[#1\right]}

\newcommand{\lrcurly}[1]{\left\{#1\right\}}

\newcommand{\Ft}[1]{\mathcal{F}_{#1}}

\newcommand{\Ldpshort}[2]{L^{#1}_{#2}}

\newcommand{\LdzF}[1]{\Ldpshort{}{#1}}
\newcommand{\LdpK}[3]{\Ldpshort{}{#2}(#3)}

\newcommand{\Lp}[2]{L^{#1}({#2})}
\newcommand{\Lpshort}[2]{L^{#1}_{#2}}

\newcommand{\LzF}[1]{\Lpshort{}{#1}(\R)}
\newcommand{\LpK}[3]{\Lpshort{}{#2}(#3)}

\newcommand{\E}[1]{\mathbb{E}\lrsquare{#1}}

\newcommand{\Et}[2]{\E{\left.#1 \right| \mathcal{F}_{#2}}}

\newcommand{\seq}[1]{(#1_t)_{t \in \T}}
\newcommand{\seqone}[1]{(#1_t)_{t \in \T\backslash \{0\}}}

\newcommand{\trans}[1]{#1^{\mathsf{T}}}
\newcommand{\transp}[1]{\trans{\lrparen{#1}}}

\newcommand{\diag}[1]{\operatorname{diag}\lrparen{#1}}
\newcommand{\cl}{\operatorname{cl}}
\newcommand{\co}{\operatorname{co}}

\newcommand{\interior}{\operatorname{int}}

\renewcommand{\succ}{\operatorname{succ}}

\newcommand{\as}{\text{a.s.}}

\newcommand{\1}{\mathbf{1}}

\newcommand{\K}{{\bf K}}
\newcommand{\XX}{\mathcal{X}}
\newcommand{\ZZ}{\mathcal{Z}}

\begin{document}
\title{A recursive algorithm for multivariate risk measures and a set-valued Bellman's principle}
\author{Zachary Feinstein \thanks{Washington University in St. Louis, Department of Electrical and Systems Engineering, St. Louis, MO 63108, USA, zfeinstein@ese.wustl.edu.} \and Birgit Rudloff \thanks{Vienna University of Economics and Business, Institute for Statistics and Mathematics, Vienna A-1020, AUT, birgit.rudloff@wu.ac.at.}}
\maketitle
\abstract{
A method for calculating multi-portfolio time consistent multivariate risk measures in discrete time is presented.  Market models
for $d$ assets with transaction costs or illiquidity and possible trading constraints are considered on a finite probability space.  The set of capital requirements at
each time and state is calculated recursively backwards in time along the event tree.  We motivate why the proposed procedure can be seen as a set-valued Bellman's principle, that might be of independent interest within the growing field of set optimization. We give conditions under which the backwards calculation of the sets reduces to solving a sequence of linear, respectively convex vector optimization problems. Numerical examples are given and include superhedging under illiquidity, the set-valued entropic risk measure, and the multi-portfolio time consistent version of the relaxed worst case risk measure and of the set-valued average value at risk.
\\[.2cm]
\\[.2cm]
}

\section{Introduction}
\label{sec_intro}

Multivariate risk measures, also called set-valued risk measures, were introduced in a static one-period setting by Meddeb, Jouini, Touzi~\cite{JMT04} in 2004, and further studied in~\cite{HR08,HH10,HHR10,CM13}, to consider risk evaluations of random vectors motivated by market models with transaction costs.
Dynamic set-valued risk measures were presented in \cite{FR12,FR12b,FR13-survey,TL12}.
The set-valued version of time-consistency, called multi-portfolio time consistency, was introduced in~\cite{FR12}
and it was shown to be equivalent to a recursive form for the risk measure.

In the present paper we will show that this recursive form can be seen as a set-valued version of Bellman's principle. On one hand, it enables us to calculate the
value of a risk measure, that is, the set of all risk compensating initial portfolio holdings, backwards in time. This is in the spirit of dynamic programming.  On the other hand, one can show that the principle of optimality holds true: the truncated optimal strategy calculated at time $t=0$ is still optimal for the optimization problems appearing at any later time point $t>0$.

We will provide conditions under which the recursive form is equivalent to a sequence of linear vector optimization problems, which can be solved with Benson's algorithm, see e.g.~\cite{B98,L11,HLR13,ELS12}. This is the case for most coherent, but also for convex polyhedral risk measures. More generally, we will give conditions under which the recursive form is equivalent to a sequence of convex vector optimization problems, that can be approximately solved by the algorithms proposed in \cite{LRU13}. Numerical examples will illustrate the results.
We will give a financial interpretation of the optimal strategies, which provides an intuition for the nested formulation of the risk measures.

The organization of this paper is as follows.  In section~\ref{sec_dynamic} we will give a short introduction to dynamic set-valued risk measures and multi-portfolio time consistency based on \cite{FR12,FR12b}.  
The main results of this paper are in sections~\ref{sec_alg} and~\ref{sec_computation}, where we reformulate the problem in terms of vector optimization. We will discuss the linear case in section~\ref{sec_polyhedral}, and the case of convex vector optimization problems in section~\ref{sec_convex}. In section~\ref{sec_bellman} we give interpretations and the link to Bellman's principle of optimality. In section~\ref{sec_market}, we study the calculation of market extensions of set-valued dynamic risk measures that appear when trading opportunities at the market are taken into account. Interpretations of the results will be given.
Numerical examples will be studied in section~\ref{sec_ex}, including superhedging, the relaxed worst case risk measure, the average value at risk, and the entropic risk measure.

\section{Set-valued dynamic risk measures}
\label{sec_dynamic}
Consider a discrete time setting $\T = \{0,1,...,T\}$ with finite time horizon $T$ and a finite filtered probability space $(\Omega,\Ft{},\seq{\mathcal{F}},\P)$ with the power set sigma algebra, i.e. $\Ft{} = 2^{\Omega}$, and $\Ft{T} = \Ft{}$. 
Denote by $\LdzF{t} = \Lp{0}{\Omega,\Ft{t},\P;\R^d}$ the linear space of the equivalence classes of $\Ft{t}$-measurable functions $X: \Omega \to \R^d$ and by $\LdzF{} := \LdzF{T}$.

We write $\LdzF{t,+} = \{X \in \LdzF{t}: \; X \in \R^d_+ \; \P\text{-}\as\}$ for the closed convex cone of $\R^d$-valued $\Ft{t}$-measurable
random vectors with non-negative components.  Note that an element $X \in \LdzF{t}$ has components $X_1,...,X_d$ in
$\LzF{t} = \Lp{0}{\Omega,\mathcal F_{t},\P;\R}$. More generally, we denote by $\LdpK{0}{t}{D_t}$ those random vectors (or variables) in $\LdzF{t}$ (resp. $\LzF{t}$) that take $\P$-a.s.
values in $D_t$.

As in \cite{K99} and discussed in \cite{S04,KS09}, the portfolios in this paper are in `physical units' of an asset rather than the value in a fixed num\'{e}raire.  That is, for a portfolio $X \in \LdzF{t}$, the values of $X_i$ (for $1 \leq i \leq d$) are the number of units of asset $i$ in the portfolio at time $t$.

Let the set of eligible portfolios $M\subseteq \R^d$, i.e. those portfolios which can be used to compensate for the risk of a portfolio, be the same for all times $t$ and states $\omega\in\Omega$.
Let $M$ be a finitely generated linear subspace of $\R^d$ with $M_+:=M\cap\R^d_+ \neq \{0\}$. Denote $M_t :=\LdpK{0}{t}{M}$ 
and $M_{t,+} =\LdpK{0}{t}{M_+}$.  
Of particular interest, especially when dealing with the market extension discussed in section~\ref{sec_market}, is the case where all assets are eligible, i.e., $M = \R^d$.

A conditional risk measure is a function which maps a $d$-dimensional random variable into $\mathcal{P}(M_t;M_{t,+}) := \{D \subseteq M_t: D = D + M_{t,+}\}$, a subset of the power set of $M_t$. 
That is, conditional risk measures map into collections of random vectors in $M_t$.
Conditional risk measures were defined in \cite{FR12} as follows.

A function $R_t: \LdzF{} \to \mathcal{P}(M_t;M_{t,+})$ is a \textbf{\emph{conditional risk measure}} at time $t$ if it is
\begin{enumerate}
\item \label{translative} $M_t$-translative: $\forall m_t \in M_t: R_t(X + m_t) = R_t(X) - m_t$;
\item \label{monotone} $\LdzF{+}$-monotone: $Y-X \in \LdzF{+} \Rightarrow R_t(Y) \supseteq R_t(X)$;
\item \label{finite-0} finite at zero: $\emptyset \neq R_t(0) \neq M_t$.
\end{enumerate}

A conditional risk measure $R_t$ is called \textbf{\emph{normalized}} if  \[R_t(X) = R_t(X) + R_t(0)\] holds for every $X \in \LdzF{t}$.
It is \textbf{\emph{local}} if for every $X \in \LdzF{}$ and every $A \in \Ft{t}$ it holds \[1_A R_t(X) = 1_A R_t(1_A X),\] where $1_A$ is the indicator function. 
A conditional risk measure at time $t$ is called \textbf{\emph{conditionally convex}} if for all $X,Y \in \LdzF{}$, and for all $\lambda \in \LpK{\infty}{t}{[0,1]}$
\[R_t(\lambda X + (1-\lambda)Y) \supseteq \lambda R_t(X) + (1-\lambda)R_t(Y),\]
it is \textbf{\emph{conditionally positive homogeneous}} if for all $X \in \LdzF{}$ and for all $\lambda \in \LpK{\infty}{t}{\R_{++}}$
\[R_t(\lambda X) = \lambda R_t(X),\] and it is \textbf{\emph{conditionally coherent}} if it is conditionally convex and conditionally positive homogeneous.
A conditional risk measure $R_t$ is \textbf{\emph{closed}} if $\operatorname{graph}(R_t) := \{(X,u) \in \LdzF{} \times M_t: u \in R_t(X)\}$ is closed in the product topology.
It is called \textbf{\emph{(conditionally) convex upper continuous}} if \[R_t^{-1}(D) := \lrcurly{X \in \LdzF{}: R_t(X) \cap D
\neq \emptyset}\] is closed for any closed (conditionally) convex set $D \in \mathcal{G}(M_t;M_{t,-}) := \{D \subseteq M_t: D = \cl\co(D + M_{t,-})\}$.

A \textbf{\emph{dynamic risk measure}} is a sequence of conditional risk measures $\seq{R}$. It is said to have one of the above properties if $R_t$ has this property for every $t \in \T$.
The \textbf{\emph{acceptance set}} at time $t$ associated with a conditional risk measure $R_t$ is defined by \[A_t = \lrcurly{X \in \LdzF{}: 0 \in R_t(X)}.\]

Let us now discuss the economic interpretation of a  dynamic risk measure $\seq{R}$. At time $t=0$, one would choose a portfolio $u_0\in R_0(X)$, usually to be as small as possible, that is, a weakly minimal element of the set $R_0(X)$, to be put aside to make $X$ acceptable at time $T$ according to the acceptance set $A_0$. As time progresses to $t=1$ and new information become available, one would update this risk compensating portfolio to keep the overall position acceptable (now according to $A_1$) by choosing a portfolio $u_1\in-u_0+ R_1(X)$, again usually a weakly minimal element of this set. This could mean injecting more capital/assets or extracting them. Now in total $u_0+u_1\in R_1(X)$ has been put aside to compensate the risk of $X$. This procedure continues until time $T$. Thus, at any time $t\in\T$ one has put aside a portfolio that compensates for the risk of $X$ according to the time $t$ risk measure $R_t$.

To ensure that updating the risk evaluation is done in a time consistent way, the concept of time consistency for scalar risk measures was extended to the set-valued framework in \cite{FR12} and is called multi-portfolio time consistency.
\begin{definition}
\label{defn_mptc}
A dynamic risk measure $\seq{R}$ is called multi-portfolio time consistent if for all times $t \in \T \backslash \{T\}$, all portfolios $X \in \LdzF{}$, and all sets ${\bf Y} \subseteq \LdzF{}$ the implication
\begin{equation}
R_{t+1}(X) \subseteq   \bigcup_{Y \in {\bf Y}} R_{t+1}(Y) \Rightarrow R_t(X) \subseteq \bigcup_{Y \in {\bf Y}} R_t(Y)
\end{equation}
is satisfied.
\end{definition}

In \cite[theorem 3.4]{FR12}, it was shown that
for a normalized dynamic risk measure $\seq{R}$, with $R_t: \LdzF{} \to \mathcal{P}(M_t;M_{t,+})$ for all times $t$,  $\seq{R}$ being multi-portfolio time consistent is equivalent to
$\seq{R}$ being recursive, that is for every time $t \in \T \backslash \{T\}$
\begin{equation}
\label{recursive}
R_t(X) =  \bigcup_{Z \in R_{t+1}(X)} R_{t,t+1}(-Z)=:R_{t,t+1}(-R_{t+1}(X)),
\end{equation}
where $R_{t,t+1}: M_{t+1} \to \mathcal{P}(M_t;M_{t,+})$ denotes the stepped risk measure $R_{t,t+1}=R_t|_{M_{t+1}}$, that is the restriction of $R_t$ to $M_{t+1}$.

Given an arbitrary dynamic risk measure $\seq{R}$ on $\LdzF{}$, one can compose the one-stepped risk measures $R_{t,t+1}$ backwards in time to obtain a multi-portfolio time consistent risk measure $\seq{\tilde{R}}$ as follows: For all $X \in \LdzF{}$ define
\begin{align}
\label{eqn_composed_final} \tilde{R}_{T}(X) & = R_{T}(X),\\
\label{eqn_composed} \forall t \in \lrcurly{0,1,...,T-1}: \quad \tilde{R}_t(X) & = \bigcup_{Z \in \tilde{R}_{t+1}(X)} R_{t,t+1}(-Z).
\end{align}
Then, $\seq{\tilde{R}}$ is multi-portfolio time consistent and satisfies the properties of $M_t$-translativity and monotonicity of a dynamic risk measures, but may fail to be finite at zero.  Additionally, if $\seq{R}$ is (conditionally) convex, (conditionally) coherent, or (conditionally) convex upper continuous and (conditionally) convex,
then $\seq{\tilde{R}}$ has the same property, see proposition 3.11 in~\cite{FR12} and proposition~5.1 in~\cite{FR12b}.

\begin{example}
\label{sec_avar}
The set-valued average value at risk was introduced and computed in~\cite{HRY12} in the static setting.  In~\cite{FR12} the definition was extended to the
dynamic framework as follows. For parameters $\lambda^t \in \LdzF{t}$, where $0 < \lambda^t_i < 1$ and $X\in \LdzF{}$ define
\begin{equation*}
AV@R_t^{\lambda}(X) = \{\diag{\lambda_t}^{-1}  \Et{Z}{t} - z :
    Z \in\LdzF{+}, \; X + Z - z \in \LdzF{+}, \; z \in \LdzF{t}\} \cap M_t.
\end{equation*}
$\seq{AV@R^{\lambda}}$ is a normalized closed conditionally coherent dynamic risk measure. The multi-portfolio time consistent version
$\widetilde{AV@R}_t^\lambda$ was discussed in~\cite{FR12b}, with the dual representation deduced for the case $M = \R^d$.  Since we are
only considering a finite probability space, we can immediately conclude (as in~\cite{HRY12})
that the dynamic average value at risk is a polyhedral risk measure.  
\end{example}

\section{A set-valued Bellman's principle}
\label{sec_alg}

We now want to answer the question if it is possible to use the nested formulation \eqref{recursive}, or, more generally, the backward composition \eqref{eqn_composed_final}, \eqref{eqn_composed} to explicitly calculate the set $R_t(X)$, respectively the multi-portfolio time consistent version $\tilde{R}_t(X)$, backwards in time. If this is possible it would justify calling this procedure a {\bf set-valued Bellman's principle}, yielding a dynamic programming method for set-valued functions. This would be an interesting insight in itself within the field of set-optimization with applications beyond the one considered here.

Recall that we assumed a finite sample space $\Omega$ with the power set sigma algebra, i.e. $\Ft{} = 2^{\Omega}$.  We define $\Omega_t$ as the set of atoms in $\Ft{t}$. For any $\omega_t \in \Omega_t$ we denote the successor nodes by
\[\succ(\omega_t) = \{\omega_{t+1} \in \Omega_{t+1}: \omega_{t+1} \subseteq \omega_t\}.\]
We use the convention that for an $\Ft{t}$-measurable random variable $u$, we denote by $u(\omega_t)$ the value of $u$ at node $\omega_t$, that is $u(\omega_t):=u(\omega)$ for some $\omega\in\omega_t$.
Further, we denote by $R_t(X)[\omega_t] := \{u(\omega_t): u \in R_t(X)\}$ the collection of projections of elements of $R_t(X)$
onto $\omega_t$. Though $R_t(X)$ is a collection of random variables rather than a random set, $R_t(X)[\cdot]$ is a random set.

In order to study a possible calculation of a set of random variables $\tilde{R}_t(X)$ backwards in time on a finite event tree, one first needs to check if one can calculate the set $\tilde{R}_t(X)$ $\omega_t$-wise at each node. That is, we wish to verify that 
\begin{align*}
	u\in \tilde R_t(X) \; \iff \; u(\omega_t)\in \tilde R_t(X)[\omega_t] \quad\forall \omega_t\in\Omega_t.
\end{align*}
This consideration is not a concern when dealing with scalar risk measures, but in the set-valued case certain conditions are needed to ensure it.
Since we work on a finite probability space, for an $\omega_t$-wise approach to \eqref{eqn_composed_final},
\eqref{eqn_composed} one will need $\seq{\tilde R}$ to have $\Ft{t}$-decomposable images, 
i.e., \[\tilde R_t(X) = 1_A \tilde R_t(X) + 1_{A^c} \tilde R_t(X) \quad \forall X \in \LdzF{} \; \forall A \in \Ft{t},\]
which is satisfied if $\seq{\tilde R}$ has 
conditionally convex images.
Furthermore, for the multi-portfolio time consistent version to only depend on the possible future (successor) states, one will
need $\seq{R}$ to be local.
\begin{remark}
Assuming $\seq{R}$ to be conditionally convex implies both, $\seq{R}$ being local (see proposition~2.8 in \cite{FR12}), as well as $\seq{\tilde R}$ being conditionally convex (proposition 3.11 in~\cite{FR12}) and thus having conditionally convex images.
\end{remark}
In order to implement the set optimization problem, we wish to reframe the backward composition~\eqref{eqn_composed_final}, \eqref{eqn_composed} as a vector optimization problem.  For this reason we require that $\seq{\tilde R}$ has closed images as well.
\begin{remark}
It is challenging to ensure that  $\seq{\tilde R}$ has closed images.
Let us give three examples, where $\seq{\tilde R}$ is closed (and thus has closed images): a)  if the dynamic risk measure $\seq{R}$ is 
convex upper continuous and convex (see proposition~5.1 in~\cite{FR12b}), or b) if the dynamic risk measure $\seq{R}$ is conditionally convex upper continuous and conditionally convex, or c) if the dynamic risk measure $\seq{R}$ is polyhedral (that is if $\operatorname{graph}(R_t)$ is a convex polyhedron, i.e. the intersection of finitely many closed half-spaces).
\end{remark}

If $\seq{\tilde{R}}$ defined in \eqref{eqn_composed_final}, \eqref{eqn_composed} is conditionally convex, but does not already have closed 
images, 
one needs to consider its closed-valued version, i.e.
\begin{align}
\label{T}
\bar{R}_T(X) &:= \cl (R_T(X))\\
\label{t}
\forall t \in \lrcurly{0,1,...,T-1}: \quad  \bar{R}_t(X) &:= \cl \bigcup_{Z \in \bar{R}_{t+1}(X)} R_t(-Z)
\end{align}
for all portfolios $X \in \LdzF{}$.
We will show that $\seq{\tilde{R}}$ can be approximated for arbitrarily small $\delta > 0$ by $\seq{\bar{R}}$ that admits an $\omega_t$-wise representation. The approximation is understood in the following sense.
$\seq{\bar{R}}$ is called an \textbf{\emph{approximation}} of $\seq{\tilde{R}}$ if
\[
	\bar{R}_t(X) + \delta m \1  \subseteq \tilde{R}_t(X) \subseteq \bar{R}_t(X)
\]
for any time $t$, for every portfolio $X \in \LdzF{}$, for any approximation tolerance $\delta > 0$ and any $m \in \interior(M_+)$ in
the subspace topology, i.e. with the topology $\tau_M := \{A \cap M: A \in \tau\}$ where $\tau$ is the topology on $\R^d$.
Thus, if $\seq{\tilde R}$ does not have closed images, we will use construction \eqref{T}, \eqref{t} to calculate an approximation $\seq{\bar{R}}$ of $\seq{\tilde{R}}$.

\begin{theorem}
\label{thm_alg2}
Let $\seq{R}$ be a conditionally convex dynamic risk measure. 
Let $\seq{\tilde{R}}$  denote its multi-portfolio
time consistent version as defined in \eqref{eqn_composed_final}, \eqref{eqn_composed}.  Then, we can calculate an approximation $\seq{\bar{R}}$ of $\seq{\tilde{R}}$ in an $\omega_t$-wise manner by
\begin{align}
\label{eq_alg-approx-T}
\bar{R}_T(X)[\omega_T] &= \cl (R_T(X)[\omega_T]),\quad \forall\omega_T \in \Omega_T;\\
\label{eq_alg-approx-t}
\bar{R}_t(X)[\omega_t] =& \cl \bigcup \lrcurly{R_{t,t+1}(-Z)[\omega_t]: Z(\omega_{t+1}) \in
    \bar{R}_{t+1}(X)[\omega_{t+1}] \; \forall \omega_{t+1} \in \succ(\omega_t)}
\end{align}
for every state $\omega_t \in \Omega_t$ and all times $t \in \T \backslash \{T\}$.
\end{theorem}
\begin{proof}
Define $\seq{\bar{R}}$ by equations \eqref{T}, \eqref{t}. Clearly,  $\seq{\bar{R}}$  has closed images.
One can show that $\seq{\bar{R}}$ is conditionally
convex by backward induction. The assertion is clearly true for $\bar{R}_T$. Now assume $\bar{R}_{t+1}$ is conditionally
convex and let $X,Y \in \LdzF{}$ and $\lambda \in \LpK{\infty}{t}{[0,1]}$. Then,
\begin{align*}
\lambda \bar{R}_t(X) &+ (1-\lambda) \bar{R}_t(Y) = \lambda \cl \bigcup_{Z_X \in \bar{R}_{t+1}(X)} R_t(-Z_X) + (1-\lambda) \cl \bigcup_{Z_Y \in \bar{R}_{t+1}(Y)} R_t(-Z_Y)\\
&\subseteq \cl \lrparen{\cl \bigcup_{Z_X \in \bar{R}_{t+1}(X)} \lambda R_t(-Z_X) + \cl \bigcup_{Z_Y \in \bar{R}_{t+1}(Y)} (1-\lambda) R_t(-Z_Y)}\\
&= \cl \bigcup_{\substack{Z_X \in \bar{R}_{t+1}(X)\\ Z_Y \in \bar{R}_{t+1}(Y)}} \lrsquare{\lambda R_t(-Z_X) + (1-\lambda) R_t(-Z_Y)}\\
&\subseteq \cl \bigcup_{\substack{Z_X \in \bar{R}_{t+1}(X)\\ Z_Y \in \bar{R}_{t+1}(Y)}} R_t(-(\lambda Z_X + (1-\lambda) Z_Y))\\
&= \cl \bigcup_{Z \in \lambda \bar{R}_{t+1}(X) + (1-\lambda) \bar{R}_{t+1}(Y)} R_t(-Z)\\
&\subseteq \cl \bigcup_{Z \in \bar{R}_{t+1}(\lambda X + (1-\lambda) Y)} R_t(-Z) = \bar{R}_{t}(\lambda X + (1-\lambda) Y).
\end{align*}
Since the probability space is assumed to be finite, $\Omega_t$ is by definition the finest partition of $\Omega$ in $\Ft{t}$ and $1_{\omega_t}u \in 1_{\omega_t}\bar R_t(X)$ if and only if $u(\omega_t) \in \bar R_t(X)[\omega_t]$.
Since $\seq{\bar{R}}$ has 
conditionally convex images and the probability space is finite, it follows that $u \in \bar{R}_t(X)$ if and only if $u(\omega_t) \in \bar{R}_t(X)[\omega_t]$ for every $\omega_t \in \Omega_t$. Thus, one can calculate $\bar{R}_t(X)$ $\omega_t$-wise.
Therefore the terminal condition \eqref{eq_alg-approx-T} holds trivially.
Let $t \in \T \backslash \{T\}$ and $\omega_t \in \Omega_t$, using \eqref{t} and the local property for $R_t$ (see~\cite[proposition 2.8]{FR12}) it follows that
\begin{align*}
1_{\omega_t} \bar{R}_t(X) &= 1_{\omega_t} \cl \bigcup_{Z \in \bar{R}_{t+1}(X)} R_t(-Z)\\ 
&= \cl\bigcup_{Z \in\bar{R}_{t+1}(X)} 1_{\omega_t} R_t(1_{\omega_t}(-Z)) = \cl\bigcup_{Z \in 1_{\omega_t}\bar{R}_{t+1}(X)} 1_{\omega_t}R_t(-Z).
\end{align*}
Note that $Z \in 1_{\omega_t}\bar{R}_{t+1}(X)$ if and only if $Z(\omega_{t+1}) \in\bar{R}_{t+1}(X)[\omega_{t+1}]$ for every
$\omega_{t+1} \in \succ(\omega_t)$ and $Z(\omega_{t+1}) = 0$ otherwise.  But as we only need to consider $1_{\omega_t}Z$ by $R_t$ local, the only constraints on $Z$ are imposed by $\omega_{t+1} \in \succ(\omega_t)$.  Thus, \eqref{eq_alg-approx-t} follows.

Finally we will show that $\bar{R}_t$ is an approximation of $\tilde{R}_t$.  By definition, $\bar{R}_{T}(X) := \cl (\tilde{R}_{T}(X))$ for all $X \in \LdzF{}$, which implies that $\bar{R}_{T}(X) + \delta_{T} m \1 \subseteq \tilde{R}_{T}(X) \subseteq \bar{R}_{T}(X)$ for any portfolio $X \in \LdzF{}$, for any $\delta_{T} > 0$, and any $m \in \interior(M_+)$.  For the proof by induction, assume that $\bar{R}_{t+1}$ is an approximation of $\tilde{R}_{t+1}$. Then  for $t \leq T-1$ we have
\begin{align*}
\tilde{R}_t(X) &= \bigcup_{Z \in \tilde{R}_{t+1}(X)} R_t(-Z) \subseteq \cl \bigcup_{\bar{Z} \in \bar{R}_{t+1}(X)} R_t(-\bar{Z}) = \bar{R}_t(X)\\
&\subseteq \cl \bigcup_{Z \in \tilde{R}_{t+1}(X)} R_t(-Z + \delta_{t+1}m \1 ) = \cl \bigcup_{Z \in \tilde{R}_{t+1}(X)} R_t(-Z) - \delta_{t+1} m \1 \\
&\subseteq \bigcup_{Z \in \tilde{R}_{t+1}(X)} R_t(-Z) - (\epsilon + \delta_{t+1})m \1 = \tilde{R}_t(X) - (\epsilon + \delta_{t+1}) m\1
\end{align*}
for any $\epsilon,\delta_{t+1} > 0$, and any $m \in \interior(M_+)$. The first inclusion on the second line follows from the induction hypothesis.
The first inclusion on the third line follows since $\cl (\tilde{R}_t(X)) \subseteq \tilde{R}_t(X) - \epsilon m$ for any $\epsilon > 0$
and all $X \in \LdzF{}$. Denote $\delta_t := \epsilon + \delta_{t+1} > 0$.   Therefore for any time $t$ and any $\delta_t > 0$ we have
that $\bar{R}_t(X) + \delta_t m\1 \subseteq \tilde{R}_t(X) \subseteq \bar{R}_t(X)$.
\end{proof}

In the recursion \eqref{eq_alg-approx-t}, the calculation is dependent on $R_{t,t+1}(-Z)[\omega_t]$ which is by locality of $R_t$ equal to $R_{t,t+1}(-1_{\omega_t}Z)[\omega_t]$, i.e. $R_{t,t+1}(-Z)[\omega_t]$ only depends on the part of $Z$ that can be attained from state $\omega_t$.
For a local risk measure, we can therefore define $R_{t,t+1}(\cdot)[\omega_t]$ on $M_{t+1}[\omega_t] := \Lp{0}{\cup\succ(\omega_t),2^{\succ(\omega_t)},\P(\cdot | \omega_t);M}$
(the equivalence class of $\Ft{t+1}$-measurable random variables in the eligible space $M$ starting from node $\omega_t$)
by $R_{t,t+1}(Z)[\omega_t] := R_{t,t+1}(\hat{Z})[\omega_t]$ for
$Z \in M_{t+1}[\omega_t]$ and some $\hat{Z} \in M_{t+1}$ with $Z(\omega_{t+1}) =
\hat{Z}(\omega_{t+1})$ for every $\omega_{t+1} \in \succ(\omega_t)$.

\begin{remark}
If $\seq{\tilde{R}}$ defined in \eqref{eqn_composed_final}, \eqref{eqn_composed} does have closed images already, $\seq{\tilde{R}}$ coincides with $\seq{\bar{R}}$ and thus can be calculated in an $\omega_t$-wise manner by \eqref{eq_alg-approx-T},~\eqref{eq_alg-approx-t}.
\end{remark}

Observe that \eqref{t} is a set-valued optimization problem in the complete lattice $\mathcal{G}(M_t;M_{t,+}):= \{D \subseteq M_t: D = \cl\co(D + M_{t,+})\}$,
and \eqref{eq_alg-approx-t} is a set-valued optimization problem in the complete lattice $\mathcal{G}(M;M_+)$.  As such, while the initial proposed problem \eqref{t} requires a lattice in sets of random vectors, the $\omega_t$-wise representation allows us to consider the lattice in sets of real-valued vectors instead.  That is, the objective
function $R_{t,t+1}$ at node $\omega_t$ is a set-valued function that is minimized over the constraint set $\bar{\ZZ}_t := \{Z\in M_{t+1}[\omega_t]:  Z(\omega_{t+1}) \in \bar{R}_{t+1}(X)[\omega_{t+1}] \; \forall \omega_{t+1} \in \succ(\omega_t)\}$:
\begin{align}
\label{setOP}
	\bar{R}_t(X)[\omega_t] =\cl\bigcup_{Z\in \bar{\ZZ}_t} R_{t,t+1}(-Z)[\omega_t]=\inf_{Z\in \bar{\ZZ}_t} R_{t,t+1}(-Z)[\omega_t].
\end{align}
Recall from \cite{L11} that the (lattice) infimum over a function $f: M_{t+1}[\omega_t] \to\mathcal{G}(M;M_+)$ is given by
$\inf_{X\in\XX} f(X)=\cl\co\bigcup_{X\in\XX} f(X)$ for $\XX\subseteq M_{t+1}[\omega_t]$, where the convex hull in front of the union can be dropped here as $\bar{R}_t(X)[\omega_t]$ is convex by theorem~\ref{thm_alg2}.
Using the idea of \cite{LS13}, one can transform the set-valued problem \eqref{setOP} into a linear vector optimization problem that can be solved e.g. by Benson's algorithm if the risk measure $\seq{R}$ is polyhedral, i.e., the graph of $R_t$ is a convex polyhedron. More generally, if $\seq{R}$ is the upper image of a convex vector optimization problem, one can transform the set-valued problem \eqref{setOP} into a convex vector optimization problem that can be approximately solved by the algorithms discussed in \cite{LRU13}. In both cases, one uses that the value of the set-optimization problem \eqref{setOP} can be written as the value of a vector optimization problem
\begin{align}
\label{VOP}
	\bar{R}_t(X)[\omega_t] =\inf_{(Z,Y)\in\ZZ_t} \Phi_t(Z,Y)
\end{align}
for the linear vector-valued function $ \Phi_t(Z,Y)=\{Y\}$, feasible set
\begin{align*}
	\ZZ_t=\{(Z,Y)&\in M_{t+1}[\omega_t]\times M: \;\\
         &Y\in R_{t,t+1}(-Z)[\omega_t], Z(\omega_{t+1}) \in \bar{R}_{t+1}(X)[\omega_{t+1}] \; \forall \omega_{t+1} \in \succ(\omega_t) \}
\end{align*}
and ordering cone $M_+$.
Let $\Phi_t(\ZZ_t)=\{\Phi_t(Z): Z\in\ZZ_t\}$ denote the image of the feasible set. The set $\cl(\Phi_t(\ZZ_t)+M_+)$ is called the upper image of the vector optimization problem \eqref{VOP}.
We now discuss the constraints by looking at two cases: the polyhedral case and the convex case.

\section{Computational procedure}
\label{sec_computation}
\subsection{Linear vector optimization and polyhedral risk measures}
\label{sec_polyhedral}

Recall that a risk measure $R_t$ is polyhedral if its graph is a convex polyhedron, i.e. the intersection of finitely many closed half-spaces. It is also equivalent to $R_t$ having a polyhedral acceptance set. For a polyhedral risk measure $\seq{R}$, problem \eqref{VOP} is a linear vector optimization problem.

\begin{proposition}
\label{propoLVOP}
If the dynamic risk measure $\seq{R}$ is conditionally convex and polyhedral, its multi-portfolio
time consistent version $\seq{\tilde{R}}$, defined in \eqref{eqn_composed_final}, \eqref{eqn_composed}, can be calculated $\omega_t$-wise, where in each node $\omega_t\in\Omega_t$, $t\in\T \backslash \{T\}$, the linear vector optimization problem  \eqref{VOP} has to be solved.
\end{proposition}
\begin{proof} The $\omega_t$-wise representation follows from theorem~\ref{thm_alg2}.
Now, let us show that problem \eqref{VOP} is a linear vector optimization problem.
By $\seq{R}$ polyhedral and since $R_{t,t+1}[\omega_t]$ maps into $\mathcal{G}(M;M_+)$, $R_{t,t+1}(-Z)[\omega_t]$ is the upper image of a linear vector optimization problem (see remark 5.1 in \cite{HLR13}), thus
\begin{align}
\label{ui}
	R_{t,t+1}(-Z)[\omega_t]=\{P_t(x)+M_+: B_tx\geq b_t\}
\end{align}	
for a vector $x=\transp{Z,z}$ that might include some auxiliary variable $z$, and for matrices $P_t$ and $B_t$ and vectors $b_t$ of appropriate dimensions. 
Then, the constraints $Y\in R_{t,t+1}(-Z)[\omega_t]$ can be equivalently written as
\[
(-Z,Y)\in\operatorname{graph} R_{t,t+1}[\omega_t] \; \iff \; \trans{\hat M}(Y-P_t(x))\geq 0,\; \; B_tx\geq b_t,
\]
where the matrix $\hat M$ contains the generating vectors of the positive dual $M_+^+$ of the ordering cone $M_+$.
Thus, these constraints are linear. To obtain linearity of the other constraints, note that
$\seq{R}$ polyhedral implies $\seq{R}$ closed (by definition) and $\seq{\tilde{R}}$ closed. To see the last implication observe
that $R_t$ is polyhedral if and only if $A_t$ is a polyhedron.  The acceptance set of $\tilde{R}_t$ is given by $\tilde{A}_t =
A_{t,t+1} + \tilde{A}_{t+1}$, see \cite[corollary~3.14]{FR12}.  $A_{t,t+1}$ is a polyhedron since $R_{t,t+1}$ is polyhedral, by
backwards recursion we assume $\tilde{A}_{t+1}$ is a polyhedron, and the sum of polyhedra is a polyhedron.  Therefore $\tilde{A}_t$ is a polyhedron, which is equivalent to $\tilde{R}_t$ being polyhedral (and thus closed as well).
Thus, in the polyhedral case,  $\seq{\bar{R}}$ coincides with $\seq{\tilde{R}}$. The linearity of the constraints $Z\in\tilde R_{t+1}(X)$ follow by induction.
The constraints from the terminal condition $Z(\omega_T)\in\tilde R_{T}(X)[\omega_T]=R_{T}(X)[\omega_T]$ are linear for all $\omega_T\in\Omega_T$ by $R_T$ polyhedral. Thus, let us assume the constraints $Z(\omega_{t+1}) \in \tilde{R}_{t+1}(X)[\omega_{t+1}]$ are linear for all $\omega_{t+1} \in \succ(\omega_t)$ for a given node $\omega_t\in\Omega_t$, then we need to show that $\tilde{R}_t(X)[\omega_t]$ is polyhedral. Since $\Omega$ is assumed to be finite, problem \eqref{VOP} is clearly a linear vector optimization problem. By \eqref{setOP},~\eqref{VOP} and $R_{t,t+1}(-Z)[\omega_t]$ mapping into $\mathcal{G}(M;M_+)$, we have  $\tilde{R}_t(X)[\omega_t] =\Phi_t(\ZZ_t)+M_+$, which is for finite $\Omega$ closed and polyhedral. Thus, $\tilde{R}_t(X)[\omega_t] $ is the upper image of the linear vector optimization problem \eqref{VOP}.
\end{proof}

Thus, if one assumes $\seq{R}$ to be a conditionally convex and polyhedral risk measure, $\tilde{R}_t(X)[\omega_t]$ is for every $t$ and $\omega_t$ the upper image of a linear vector optimization problem, which is polyhedral and can be calculated by Benson's algorithm (see \cite{HLR13}).
Then, the set $\tilde{R}_t(X)$ can be calculated backwards in time, by solving at each node a linear vector optimization problem. Note that at time $t$, the problems for each $\omega_t\in\Omega_t$ can be calculated in parallel instead of sequentially which reduces computational time. Several examples will be discussed in section~\ref{sec_ex}.

Benson's algorithm is an appropriate tool to solve the vector optimization problem \eqref{VOP} as it takes advantage of the fact that the dimension $\dim(M)$ of the image space is usually significantly smaller than the dimension $d\times|\succ(\omega_t)|+d+|z|$ of the pre-image space, where $|\succ(\omega_t)|$ denotes the number of successor nodes of $\omega_t$ and $|z|$ denotes the dimension of the auxiliary variables in~\eqref{ui}.

In practice (especially if $M$ is higher dimensional), when the number of vertices of the set $\tilde{R}_t(X)[\omega_t]$  is very high, one would calculate an approximation of $\tilde{R}_t(X)[\omega_t]$ having fewer vertices, see remark 4.10 in \cite{HLR13}. Then, for the backward recursion, one would need to know how the approximation errors accumulate over time. This will be discussed in propositions~\ref{propo3.5} and ~\ref{prop_recursive-approxa} below in a more general framework.

\begin{example}
\label{sec_rwc}
The relaxed worst case risk measure was introduced in the static framework in example 5.2 of~\cite{HLR13}.  The idea behind the relaxed worst case risk measure is to modify the worst case risk measure so that portfolios with ``small'' negative components can still be acceptable.

In this paper we consider the dynamic extension of such a risk measure.  By definition it is a polyhedral and conditionally
convex, but not conditionally coherent, risk measure; the acceptance set at time $t$ is given by
\[A_t^{RWC} = (-\varepsilon + \LdzF{+}) \cap \LdpK{0}{}{G}\]
for some level $\varepsilon \in \R^d_+$ and some finitely generated convex cone $G \supseteq \R^d_+$ and $G \neq \R^d$.  Note that if $G = \R^d_+$ or $\varepsilon = 0$ then the relaxed worst case risk measure is equivalent to the worst case risk measure.

Let $\seq{R}$ be the relaxed worst case risk measure with polyhedral acceptance set $A_t^{RWC}$. Then, by proposition~\ref{propoLVOP} one can
calculate its multi-portfolio time consistent version $\seq{\tilde R}$ $\omega_t$-wise, where in each node $\omega_t\in\Omega_t$,
$t\in\T \backslash \{T\}$, the linear vector optimization problem \eqref{VOP} has to be solved.  It should be noted that in general $R_t
\neq \tilde{R}_t$, i.e. the relaxed worst case risk measure is not multi-portfolio time consistent.
\end{example}

\subsection{Convex vector optimization and conditionally convex risk measures}
\label{sec_convex}

As the upper image of a convex vector optimization problem can only be calculated by a polyhedral approximation yielding an inner as well as an outer approximation with respect to some error level $\epsilon$ (see e.g. \cite{ESS11,LRU13}), we introduce approximations of sets, respectively functions. Throughout, we fix a parameter $m \in \interior M_+$.

\begin{definition}
\label{defn_approx} Given a set $S \in \mathcal{P}(M;M_+)$ and an error level $\epsilon > 0$, we call
a set $S^{\epsilon} \in \mathcal{P}(M;M_+)$ an approximation of $S$, if
\[
	S^{\epsilon} + \epsilon m \subseteq S \subseteq S^{\epsilon}.
\]
Given a set-valued function $F: \LdzF{} \to \mathcal{P}(M_t;M_{t,+})$ and an error level $\epsilon > 0$, we call
the function $F^{\epsilon}: \LdzF{} \to \mathcal{P}(M_t;M_{t,+})$ an approximation  of $F$ if
\[
	F^{\epsilon}(X) + \epsilon m\1 \subseteq F(X) \subseteq F^{\epsilon}(X) \text{ for every }X \in \LdzF{}.
\]
\end{definition}

In the convex case one can in general only approximately calculate the constraint set $\bar R_{t+1}(X)$ in the backward recursion \eqref{eq_alg-approx-t}, respectively \eqref{setOP}.
Let us study the robustness of the set-optimization problem \eqref{setOP} to this perturbation of the constraints.

\begin{proposition}
\label{propo3.5}
Let $\epsilon > 0$. Let $\bar R^\epsilon_{t+1}(X)[\omega_{t+1}]$ be an $\epsilon$-approximation of $\bar R_{t+1}(X)[\omega_{t+1}]$ for each $\omega_{t+1} \in \succ(\omega_t)$, then $\bar{R}_t^{\epsilon}(X)[\omega_t]  $ defined by
\begin{align}
\label{appCVOP}
	\bar{R}_t^{\epsilon}(X)[\omega_t]  := \cl\bigcup_{Z \in \bar{\ZZ}_t^\epsilon} R_{t,t+1}(-Z)[\omega_t],
\end{align}
with
\[
\bar{\ZZ}_t^\epsilon:=\{Z\in M_{t+1}[\omega_t]:  Z(\omega_{t+1}) \in
    \bar{R}^\epsilon_{t+1}(X)[\omega_{t+1}] \; \forall \omega_{t+1} \in \succ(\omega_t)\},
\]
is an  $\epsilon$-approximation of $\bar R_t(X)[\omega_t]$ defined in \eqref{eq_alg-approx-t}.
\end{proposition}

\begin{proof} The assumption implies $\bar{\ZZ}_t\subseteq \bar{\ZZ}_t^\epsilon\subseteq \bar{\ZZ}_t -\epsilon m\1$. This together with \eqref{setOP} and transitivity of $R_t$ yields
\begin{align*}
\bar{R}_t&(X)[\omega_t]  =\cl\bigcup_{Z\in \bar{\ZZ}_t} R_{t,t+1}(-Z)[\omega_t]\subseteq \cl\bigcup_{Z\in \bar{\ZZ}_t^\epsilon} R_{t,t+1}(-Z)[\omega_t]= \bar{R}_t^{\epsilon}(X)[\omega_t]
\\
&\subseteq \cl\bigcup_{Z+\epsilon m\1\in \bar{\ZZ}_t} R_{t,t+1}(-Z)[\omega_t]= \cl\bigcup_{Z\in \bar{\ZZ}_t} R_{t,t+1}(-Z)[\omega_t]-\epsilon m=\bar{R}_t(X)[\omega_t] -\epsilon m.
\end{align*}
Thus, $\bar{R}_t^{\epsilon}(X)[\omega_t] + \epsilon m \subseteq\bar{R}_t(X)[\omega_t] \subseteq \bar{R}_t^{\epsilon}(X)[\omega_t] $, i.e. $\bar{R}_t^{\epsilon}(X)[\omega_t] $ is a $\epsilon$-approximation of $\bar{R}_t(X)[\omega_t]$.
\end{proof}

Next, we discuss under which conditions problem~\eqref{appCVOP} is a convex vector optimization problem and which additional assumptions are necessary to apply the algorithm proposed in \cite{LRU13} to calculate a polyhedral approximation of the upper image of this convex vector optimization problem.

\begin{assumption}
\label{assCVOP}
\begin{itemize}
\item[a)]
Let the 
objective function in \eqref{appCVOP} be of the form $R_{t,t+1}(Z)[\omega_t]=\{\Psi_t(Z,z)+M_+: g_t(Z,z)\leq 0\}$ for an 
$M_+$-convex vector function $\Psi_t$, a component-wise convex vector function $g_t$ and a vector $z$, all of appropriate and finite dimensions.
\item[b)] Let the function $\Psi_t$ in a) be continuous, and let the feasible set $\XX:=\{(Z,z):g_t(Z,z)\leq 0\}$ satisfy
$\interior\XX\neq\emptyset$. 
\end{itemize}
\end{assumption}
Assumption~\ref{assCVOP}~a) means that the closure of $R_{t,t+1}(Z)[\omega_t]$ is itself the upper image of a convex vector optimization problem.

\begin{proposition}
\label{propoCVOP}
Let the objective function $R_{t,t+1}(Z)[\omega_t]$ in \eqref{appCVOP} satisfy assumption~\ref{assCVOP}~a) and let $\bar R^\epsilon_{t+1}(X)[\omega_{t+1}]$ be a polyhedron for each $\omega_{t+1} \in \succ(\omega_t)$. Then, problem~\eqref{appCVOP} is a convex vector optimization problem.
\end{proposition}
\begin{proof}
Similar to \eqref{VOP}, the set-valued problem \eqref{appCVOP} can be written as a vector optimization problem
\begin{align*}
	\inf_{(Z,Y)\in\ZZ_t^{\epsilon}} \Phi_t(Z,Y),
\end{align*}
by setting $\Phi_t(Z,Y)=\{Y\}$ (which is a linear vector function), defining the feasible set as $\ZZ_t^{\epsilon} =\{(Z,Y)\in M_{t+1}[\omega_t]\times M: \; Y \in R_{t,t+1}(-Z)[\omega_t], Z(\omega_{t+1}) \in
    \bar{R}^\epsilon_{t+1}(X)[\omega_{t+1}] \; \forall \omega_{t+1} \in \succ(\omega_t)\}$ and using $M_+$ as the ordering cone.
The constraints $Z(\omega_{t+1}) \in \bar{R}^\epsilon_{t+1}(X)[\omega_{t+1}] $ are by assumption linear.
Under assumption~\ref{assCVOP}~a), the constraints $Y\in R_{t,t+1}(-Z)[\omega_t]$ in \eqref{VOP} can be equivalently written as
\[
(Y,-Z)\in\operatorname{graph} R_{t,t+1}[\omega_t] \; \iff \; \trans{\hat M}(\Psi_t(-Z,z)-Y)\leq 0,\; \; g_t(-Z,z)\leq 0,
\]
where the matrix $\hat M$ contains the generating vectors of $M_+^+$.  $\trans{\hat M}(\Psi_t(-Z,z)-Y)$ is a component-wise convex vector function since $\Psi_t$ is a $M_+$-convex vector function. Thus, these are convex constraints and \eqref{appCVOP} is a convex vector optimization problem.
\end{proof}

The additional assumptions~\ref{assCVOP}~b) are necessary to ensure that problem~\eqref{appCVOP} can be (approximately) solved by the algorithms presented in \cite{LRU13}.
In detail, under assumptions~\ref{assCVOP} and if the feasible set $\XX:=\{(Z,z):g_t(Z,z)\leq 0\}$ is compact,  \cite[theorems~4.9 and~4.14]{LRU13} state that the algorithms in  \cite{LRU13} provide an approximation of the upper image of \eqref{appCVOP}, i.e. a polyhedral approximation of  $\bar R_t(X)$, if they terminate. However, the compactness assumption is typically not satisfied in the setting of risk measures. In that case  \cite[remark~3 in section~4.3]{LRU13} shows that the algorithms presented in \cite{LRU13} still return an approximation of the upper image of \eqref{appCVOP} as long as all the scalar optimization problems within the algorithm can be solved and the algorithm terminates.
In the example of the set-valued entropic risk measure considered in section~\ref{sec_entropic}, this will indeed be the case.

Since in general problem~\eqref{appCVOP}  can only be solved approximately (e.g. by the algorithms in \cite{LRU13}), one also need to study how the approximation errors made at different time points accumulate over time.
\begin{proposition}
\label{prop_recursive-approxa}
Let $\epsilon,\gamma > 0$.
If $\bar{R}_t^{\epsilon,\gamma}(X)[\omega_t] $ is a $\gamma$-approximation of $\bar{R}_t^{\epsilon}(X)[\omega_t] $ defined in \eqref{appCVOP}, then $\bar{R}_t^{\epsilon,\gamma}(X)[\omega_t] $ is an $(\epsilon+\gamma)$-approximation of $\bar{R}_t(X)[\omega_t] $ defined in \eqref{eq_alg-approx-t}. 
\end{proposition}
\begin{proof}
$\bar{R}_t^{\epsilon,\gamma}(X)[\omega_t] $ being a $\gamma$-approximation of $\bar{R}_t^{\epsilon}(X)[\omega_t] $ means \[\bar{R}_t^{\epsilon,\gamma}(X)[\omega_t]  + \gamma m \subseteq \bar{R}_t^{\epsilon}(X)[\omega_t] \subseteq \bar{R}_t^{\epsilon,\gamma}(X)[\omega_t].\]  Proposition~\ref{propo3.5} shows that $\bar{R}_t^{\epsilon}(X)[\omega_t] $
is an  $\epsilon$-approximation of $\bar R_t(X)[\omega_t]$, i.e. \[\bar{R}_t^{\epsilon}(X)[\omega_t]  + \epsilon m \subseteq \bar{R}_t(X)[\omega_t] \subseteq \bar{R}_t^{\epsilon}(X)[\omega_t].\] Both chains of inclusions yield \[\bar{R}_t^{\epsilon,\gamma}(X)[\omega_t]  + (\epsilon+\gamma) m \subseteq \bar{R}_t(X)[\omega_t]  \subseteq \bar{R}_t^{\epsilon,\gamma}(X)[\omega_t] .\]
\end{proof}

We are now ready to prove the main result of this section. Recall that the aim was to (approximately) calculate the multi-portfolio time consistent risk measure $\seq{\tilde R}$ backwards in time in the spirit of a set-valued Bellman's principle. We will see that $\seq{\tilde R}$ can be obtained by solving at each node backwards in time a convex vector optimization problem. In practice, these problems can only be approximately solved. But we are able to determine the overall approximation error, when the approximation error at each node is chosen to be $\epsilon > 0$. One could of course also vary this error level at different nodes or different time points and obtain corresponding results.

\begin{proposition}
\label{prop_recursive-approx}
Let $\seq{R}$ be a conditionally convex dynamic risk measure satisfying assumption~\ref{assCVOP}.  Let $\epsilon > 0$.  Then for any time $t$ and given $X \in \LdzF{}$, we can find a $[(T-t+1)\epsilon+\delta]$-approximation of the multi-portfolio time consistent version $(\tilde R_t(X))_{t\in\T}$ defined in \eqref{eqn_composed_final}, \eqref{eqn_composed}, by calculating backwards in time at each node $\omega_t\in\Omega_t$ an $\epsilon$-approximation of the upper image of the convex vector optimization problem~\eqref{appCVOP}. Here $\delta > 0$ can be chosen arbitrarily small.
\end{proposition}
\begin{proof}
Assumption~\ref{assCVOP} and the local property of $\seq{R}$ imply that all the assumptions of theorem~\ref{thm_alg2} are satisfied, thus a $\delta$-approximation $\seq{\bar{R}}$ of $\seq{\tilde{R}}$ can be calculated $\omega_t$-wise for arbitrarily small $\delta > 0$  by  \eqref{eq_alg-approx-t}, \eqref{eq_alg-approx-T}.

For $t=T$ one obtains an $\epsilon$-approximation  $\bar{R}_T^{\epsilon}(X)$ of $\bar R_T(X)=\cl(\tilde R_T(X))=\cl(R_T(X))$ by calculating an $\epsilon$-approximation of the upper image of the convex vector optimization problem $R_T(X)[\omega_T]=\{\Psi_T(X,z)+M_+: g_T(X,z)\leq 0\}$ (see assumption~\ref{assCVOP}~a)) at each node $\omega_T\in\Omega_T$. $\bar{R}_T^{\epsilon}(X)$ is by construction polyhedral (see the algorithms in \cite{LRU13}) and is the input for problem \eqref{appCVOP} at time $t=T-1$, which is by proposition~\ref{propoCVOP} then a convex vector optimization problem. Its solution would by proposition~\ref{propo3.5} yield an $\epsilon$-approximation  $\bar{R}_{T-1}^{\epsilon}(X)$ of $\bar R_{T-1}(X)$, but one can in general only calculate an $\epsilon$-solution. This $\epsilon$-solution yields an $\epsilon$-approximation  of $\bar{R}_{T-1}^{\epsilon}(X)$, which is by proposition~\ref{prop_recursive-approxa} a $2\epsilon$-approximation of $\bar{R}_{T-1}(X)$.

Going backwards like this yields for any $t$ a $(T-t+1)\epsilon$-approximation of $\bar{R}_t(X)$, which is by theorem~\ref{thm_alg2} and the logic of adding up approximation errors as in proposition~\ref{prop_recursive-approxa} a $(T-t+1)\epsilon+\delta$-approximation of the multi-portfolio time consistent version $\tilde R_t(X)$ for arbitrarily small $\delta > 0$.
\end{proof}

\begin{example}
\label{sec_entropic}
The set-valued entropic risk measure was studied in~\cite{AHR13} in a single period static framework.  The dynamic version was
discussed in~\cite{FR12b}.
As in the scalar case, the entropic risk measure is intimately related to the exponential utility function.
Consider risk aversion parameters $\lambda^t \in \LdpK{\infty}{t}{\R^d_{++}}$, $C_t \in
\mathcal{G}(\LdzF{t};\LdzF{t,+})$ with $0 \in C_t$ and $C_t \cap \LdpK{0}{t}{\R^d_{--}} = \emptyset$.
The dynamic entropic risk measure is defined by
\begin{equation}
\label{entRM}
R_t^{ent}(X;\lambda^t,C_t) := \lrcurly{u \in M_t: \Et{u_t(X + u)}{t} \in C_t }
\end{equation}
for every $X \in \LdzF{}$ where
$u_t(x) =\transp{u_{t,1}(x_1),...,u_{t,d}(x_d)}$ for any $x \in \R^d$ and $u_{t,i}(z)=\frac{1-e^{-\lambda^t_i z}}{\lambda^t_i}$ for $z\in\R$ and $i=1,...,d$.

An approximate calculation of the static entropic risk measure was shown in \cite{LRU13} via solving a convex vector optimization
problem.  With the method presented in section~\ref{sec_alg} we are able to compute an approximation $\seq{\bar{R}^{ent}}$ of the
multi-portfolio time consistent version $\seq{\tilde{R}^{ent}}$ by backward composition for a general space of eligible portfolios
$M_t$, (stochastic) risk aversion parameters $\lambda^t$, and polyhedral parameters $C_t$.  It was proven in \cite{FR12b} that the entropic risk
measure is c.u.c. and multi-portfolio time consistent in the case that $M = \R^d$, constant $\lambda \in \R^d_{++}$, and $C_t = \LdzF{t,+}$, i.e. $R_t^{ent} = \tilde{R}_t^{ent} = \bar{R}_t^{ent}$.

From the definition of the entropic risk measure with $C_t$ polyhedral, it is clear that $R_t^{ent}$ satisfies
assumption~\ref{assCVOP}. Thus, by proposition~\ref{prop_recursive-approx} one can calculate an approximation
of the multi-portfolio time consistent version $(\tilde R_t^{ent}(X))_{t\in\T}$ by calculating backwards in time at each node
$\omega_t\in\Omega_t$ an approximation of the upper image of the convex vector optimization problem~\eqref{appCVOP} using
the algorithms presented in \cite{LRU13}. 
\end{example}

\section{Interpretation and relation to Bellman's principle}
\label{sec_bellman}
The risk measure $\seq{\tilde{R}}$, while
constructed backwards in time, has a nice financial interpretation involving portfolio injections made as time progresses, that is an interpretation forwards in time:
For every choice of a risk compensating portfolio holding $Z_0\in \tilde{R}_0(X)$ at time $t=0$, there exists, by equation \eqref{eqn_composed}, a sequence of
portfolio holdings $\seqone{Z}$ such that
\begin{equation}\label{I1}
	Z_t\in \tilde{R}_t(X)
\end{equation}
and
\begin{equation}\label{I2}
Z_{t-1}\in R_{t-1}(-Z_t).
\end{equation}
Inclusion \eqref{I1} means $Z_0, Z_1,...,Z_T$ are the risk compensating portfolio holdings at times $0,1,...,T$. An intuitive interpretation of \eqref{I2} can be obtained by the following reformulation. Defining the portfolio injections
(respectively withdrawals - if negative) $\seq{u}$ that are needed to update the risk compensating portfolio holdings by $u_t=Z_t-Z_{t-1}$ (with $u_0 = Z_0$), the two conditions on $\seq{Z}$ can be rewritten in terms of $\seq{u}$ as follows
\[u_t \in \tilde{R}_t\lrparen{X + \sum_{s = 0}^{t-1} u_s}\]
for every time $t$, and
\begin{equation}\label{I3}
0\in R_{t,t+1}(-u_{t+1}),
\end{equation}
for $t\in \T\backslash \{T\}$.
Inclusion \eqref{I3} means the risk of the portfolio injection needed at time $t+1$ is acceptable at time $t$ with respect to the one-period risk measure $R_{t,t+1}$. This gives the main interpretation of the backward composition of $\seq{R}$. At each one-period step the original measure $\seq{R}$ is used, but it is used in a time consistent way in the sense of Bellman.

One can observe Bellman's principle of optimality:  The at $t$ truncated optimal solution $(Z_s)_{s=t}^T$ obtained at time $0$ from \eqref{eqn_composed} and a given $Z_0\in \tilde{R}_0(X)$ is still optimal at any later time point $t\in\T$. To see that, note that for
the risk compensating portfolio holding $Z_t\in \tilde{R}_t(X)$, $(Z_s)_{s=t}^T$ satisfies the conditions $Z_s\in \tilde{R}_s(X)$ and
$Z_{s-1}\in R_{s-1}(-Z_s)$, $s\in\{t,...,T\}$ from \eqref{eqn_composed}. 

Let us now explain on how to compute $(\tilde{R}_t(X))_{t \in \T}$ and how to obtain for a given $Z_0\in \tilde{R}_0(X)$ at time $t=0$ a sequence $\seqone{Z}$ of risk compensating portfolio holdings on the realizing path.
$\seq{\tilde{R}}$ can be calculated with the approach discussed in sections~\ref{sec_polyhedral} and~\ref{sec_convex}.  Benson's algorithm also calculates a solution of the linear vector optimization problems in the sense of definition 2.20 in \cite{L11}, respectively, an $\epsilon$-solution in the sense of definition~3.3 in \cite{LRU13} for a convex vector optimization problem. These finite
solution sets are then used to calculate the sequence of risk compensating portfolio holdings $\seq{Z}$, respectively the injection/withdrawal strategy $\seq{u}$, forwards in time on the realizing path by solving an additional linear program, specified in the following, at each point in time.
Let $\bar{\XX}_t[\omega_t] = \{(Z_{t+1}^i[\omega_t],Y_t^i[\omega_t]): i = 1,...,n, \;n\in\N\}\subseteq M_{t+1}[\omega_t] \times M$ 
be the ($\epsilon$-)solution set to the vector optimization problem~\eqref{VOP}.
Let us first explain the method in case of a linear vector optimization problem: For any $Z_0$ in the
risk measure $\tilde{R}_0(X)$, there exists a convex combination of elements of the solution on the efficient frontier (the collection of nondominated vectors with ordering $M_+$) such that 
$Z_0 \geq \sum_{i = 1}^n \lambda^*_i Y_0^i$. This coefficient vector $\lambda^* \in \R^n_+$ can be found by solving any linear optimization problem of the form
\begin{equation}
\label{eq_BellmanLP}
\min_{\lambda \in \R^n_+} \trans{c}\lrparen{Y_0^1, \cdots, Y_0^n}\lambda \quad \text{ subject to } \quad \lrparen{Y_0^1, \cdots, Y_0^n} \lambda \leq Z_0, \quad \trans{\vec{1}}\lambda = 1
\end{equation}
with $c \in \R^d_+ \backslash \{0\}$.
The coefficient vector $\lambda^* \in \R^n_+$ can then be
used to define $Z_0^* := \sum_{i = 1}^n \lambda_i^* Y_0^i$ on the efficient frontier of $\tilde{R}_0(X)$.  Notice that $Z_0^* = Z_0$ if $Z_0$ is
already on the efficient frontier. Additionally, the next time step full capital requirement is given by $Z_1:=\sum_{i = 1}^n
\lambda_i^* Z_1^i$, which might not be on the efficient frontier of $\tilde{R}_1(X)$.  This process is repeated through the event tree
forwards in time.  The choice of cost vector $c$ (or alternatively a nonlinear cost function) determines the possible
liquidation/withdrawal strategy akin to that discussed for the superhedging risk measure in~\cite{LR11}.

In the case of a convex vector optimization problem, one can calculate only a polyhedral approximation (e.g.with error level $\tilde\epsilon=(T+1)\epsilon+\delta$ as in proposition~\ref{prop_recursive-approx}) $\tilde{R}^{\tilde\epsilon}_0(X)$ of $\tilde{R}_0(X)$. Thus, when choosing the initial capital, one would pick a minimal capital from the calculated inner approximation, and not the true set, i.e. $u_0\in \tilde{R}^{\tilde\epsilon}_0(X)+\tilde\epsilon m$. Noting that an $\epsilon$-solution of problem~\eqref{VOP} provides a solution to the linear vector optimization problem whose upper image is the inner approximation, the same procedure as in the linear case can be applied, just replacing $\tilde{R}_0(X)$ by its inner polyhedral approximation. One obtains an $\tilde\epsilon_t$-optimal strategy of risk compensating portfolios $\seq{Z}$, respectively portfolio injections $\seq{u}$, with $\tilde\epsilon_t=(T-t+1)\epsilon+\delta$ when using the same error level $\epsilon>0$ in each iteration step, see proposition~\ref{prop_recursive-approx}.

\section{Notes on market extensions}
\label{sec_market}

Market extensions are considered when one is not only interested in putting a `capital requirement' $u\in R_t(X)$ at time $t$ aside and
holding it until time $T$ to make $X$ risk neutral, but in exploiting the trading opportunities at the market to minimize the amount of
capital needed for risk compensation.
For the definition of the market extension below, we will set $M = \R^d$, i.e. we consider the full space of eligible assets. A justification for that comes from a mathematical as well as an interpretational aspect, which will be detailed in remark~\ref{rem_market-mptc_M} below. But one can already understand that choice by realizing that the role of $M$ comes mainly from a regulatory point of view. A regulator might only allow capital requirements to be made in certain currencies for example, and these capital requirements are held until time $T$. But the market extension is more linked to internal risk measurement and management as one is exploring trading opportunities in possibly all assets, and thus will hold at any time $t$ a portfolio in possibly all assets, so there is no need in restricting the capital requirements to be made in certain assets only.

The \textbf{\emph{market extension}} $\seq{R^{mar}}$ of a dynamic risk measure $\seq{R}$ is given by
\[
	R_t^{mar}(X) := \bigcup_{k \in \K_t} R_t(X - k)
\]
for some $\K_t \subseteq \LdzF{}$ modeling the set of attainable claims. When $\K_t \subseteq \LdzF{t}$ then it immediately follows
that 
\begin{equation}
\label{me}
R_t^{mar}(X) = R_t(X) + \K_t. 
\end{equation}
Let us give a few examples, all are special cases of the set-valued portfolios introduced in \cite{CM13}.
\begin{example}\label{ex1}
In a market with proportional transaction costs,  trading is modeled by a sequence of solvency cones $\seq{K}$, see
\cite{K99,S04,KS09}. $K_t$ is a \textbf{\emph{solvency cone}} at time $t$ if it is an $\Ft{t}$-measurable cone such that for every
$\omega \in \Omega$, $K_t[\omega]$ is a closed convex cone with $\mathbb{R}_+^d \subseteq K_t[\omega] \subsetneq \mathbb{R}^d$. $K_t$ is generated by the bid and ask prices between any two assets at time $t$. In a market with proportional transaction costs,
one would set $\K_t=\LdpK{0}{t}{K_t}$ (see \cite{FR12}). 
\end{example}
\begin{example}
More generally, in markets with illiquidity (convex transaction costs) as in \cite{PP10}, trading is modeled by a sequence of convex solvency regions $\seq{K}$. $K_t$ is a \textbf{\emph{convex solvency region}} at time $t$ if it is an $\Ft{t}$-measurable set such that for every $\omega \in \Omega$, $K_t[\omega]$ is a closed convex set with $\mathbb{R}_+^d \subseteq K_t[\omega] \subsetneq \mathbb{R}^d$.
Then, one would set $\K_t=\LdpK{0}{t}{K_t}$.
\end{example}
\begin{example}
One could also incorporate \textbf{\emph{trading constraints}} on the size of transactions by considering convex random sets $D_t$ (not necessarily mapping into the closed convex upper sets $\mathcal G(\R^d,\R^d_+)$) as follows. Given $t\in\T$, let $D_t : \Omega\to 2^{\R^d}$ (with $2^{\R^d}$ denoting the power set of $\R^d$) be an $\Ft{t}$-measurable function such that $D_t[\omega]$ is a closed convex set and $K_t[\omega] \cap D_t[\omega] \neq \emptyset$ for every $\omega\in\Omega$.
Then, one would set $\K_t=\LdpK{0}{t}{K_t\cap D_t}$.
\end{example}
We now want to consider the market extended multi-portfolio time consistent version of $\seq{R}$ with respect to $\K_t \subseteq
\LdzF{t}$. Different possibilities arise and will be briefly discussed. First note that the market extension of the multi-portfolio time consistent version $\seq{\tilde{R}}$, given by $(\tilde{R}_t(\cdot) + \K_t)_{t \in \T}$, is not multi-portfolio time consistent, thus the solutions do not satisfy Bellman's principle as detailed in section~\ref{sec_bellman} and there is no good economic interpretation. Therefore, only carefully alternating a market extension step \eqref{me} and a backward recursion step, that is  applying \eqref{me} \emph{in each step of the backward recursion} will yield the desired results and interpretations.  As such we will define
\begin{equation}
  \label{eqn_composed2}
  \tilde{R}_t^{mar}(X) := \bigcup_{Z \in \tilde{R}_{t+1}^{mar}(X)} (R_t(-Z) + \K_t) = \bigcup_{k \in \K_t} \bigcup_{Z \in \tilde{R}_{t+1}^{mar}(X-k)} R_t(-Z).
  \end{equation}
This is the multi-portfolio time consistent version of the market extension, i.e., the backward composition of $\seq{R^{mar}}$. The obtained risk measure $(\tilde{R}_t^{mar})_{t\in\T}$ does indeed satisfy Bellman's principle and we will give the economic interpretation of the solutions below.
Equation~\eqref{eqn_composed2} shows that 
the two operations `market extension' and `multi-portfolio time consistent version' are interchangeable at any single time point $t$ within a recursion step under the assumption $M = \R^d$. However, as noted previously, the market extension of the multi-portfolio time consistent version, $(\tilde{R}_t(\cdot) + \K_t)_{t \in \T}$, is not equal to the multi-portfolio time consistent version of the market extension, $\seq{\tilde{R}^{mar}}$, defined in \eqref{eqn_composed2}.
Thus, intrinsic to the definition of $\seq{\tilde{R}^{mar}}$ is that both the market extension and backward recursion (independent of the order of operations) are computed at a time point $t+1$ and used as the input for the backward recursion in the next time point $t$.

In analogy to section~\ref{sec_bellman} for the regulator risk measure, one can obtain a nice financial interpretation of
 the market extended composed risk measure $\seq{\tilde{R}^{mar}}$ involving portfolio injections and trades made forwards in time.  In addition to the sequence of portfolios holdings $\seq{Z}$ one obtained for the regulator risk measure, there additionally exists by equation \eqref{eqn_composed2} a sequence of trades $\seq{k}$ such that $Z_t\in \tilde{R}_t^{mar}(X)$, $k_t\in \K_t$, and $Z_t-k_t\in R_t(-Z_{t+1})$ for every choice of a risk compensating portfolio holding $Z_0\in \tilde{R}_0^{mar}(X)$ at time $t=0$.  That means $Z_0, Z_1,...,Z_T$ are the risk compensating portfolio holdings before trades at times $0,1,...,T$ and $Z_0-k_0,
Z_1-k_1,...,Z_T-k_T$ are the risk compensating portfolio holdings after the trades at times $0,1,...,T$.  Equivalently, the portfolio injections
(respectively withdrawals - if negative) $\seq{u}$, needed to update the risk compensating portfolio holdings and defined by $u_t=Z_t-Z_{t-1}+k_{t-1}$ (with $u_0 = Z_0$), satisfy
\[u_t \in \tilde{R}_t^{mar}\lrparen{X + \sum_{s = 0}^{t-1}(u_s - k_s) - k_t}\]
for every time $t$, and
\[0\in R_t(-u_{t+1}),\]
for $t\in \T\backslash \{T\}$.
This gives the same interpretation as with the composed regulator risk measure discussed in section~\ref{sec_bellman}: The portfolio injections of the next time period $t+1$ are random, but acceptable with respect to the one-period risk measure $R_{t,t+1}$.

Let us now discuss how to calculate $\tilde{R}_t^{mar}(X)$. Assume $\K_t$ is closed and conditionally convex
for all times $t$, then $\seq{\tilde{R}^{mar}}$ can be calculated with the approach discussed in sections~\ref{sec_polyhedral}
and~\ref{sec_convex} by adding $\K_t[\omega_t]$ to the vector optimization problem \eqref{VOP} at each time $t$. 
If $\K_t = \LdpK{0}{t}{K_t}$ for a \emph{solvency cone} $K_t$ for all times $t$ (as in example~\ref{ex1}), the set-valued risk measure can also be computed directly by replacing the ordering cone $M_+ = \R^d_+$ with $K_t[\omega_t]$ in problem \eqref{VOP}.

As with the `regulator risk measure' considered in the previous sections, the market extension $\tilde{R}_t^{mar}(X)$ might not be closed, but an arbitrarily close approximation is given by its closed-valued variant
\[
\bar{R}_t^{mar}(X) := \cl \bigcup_{Z \in \bar{R}_{t+1}^{mar}(X)} (R_t(-Z) + \K_t).
\]
Solving at each node backwards in time the vector optimization problem with objective $\Phi_t(Z,Y + k)$ and feasible region $(Z,Y,k) \in \ZZ_t \times \K_t[\omega_t]$ (with $\bar{R}_{t+1}$ replaced by $\bar{R}_{t+1}^{mar}$) yields $\seq{\bar{R}^{mar}}$.
Benson's algorithm yields  the solution set
$\bar{\XX}_t^{mar}[\omega_t] \subseteq \LdzF{t+1}[\omega_t] \times \R^d \times \K_t[\omega_t]$ (with $\LdzF{t+1}[\omega_t] = \Lp{0}{\cup\succ(\omega_t), 2^{\succ(\omega_t)}, \P(\cdot | \omega_t);\R^d}$),
which can be used to calculate the sequence
$\seqone{Z}$ and now additionally the sequence of trades $\seq{k}$ forwards in time on the realizing path for a given $Z_0\in \bar{R}_0^{mar}(X)$ (or its inner approximation) at time $t=0$.
Utilizing \eqref{eq_BellmanLP}, we can find a convex combination of the $Y$ elements of the solution to describe any portfolio on the efficient frontier of the risk measure, respectively of the inner approximation of it in the convex, non-polyhedral case.  The same convex coefficients are then used for both the trading strategy and the next time step full capital requirements, which are then used as the starting value in the next period.

\begin{remark}
\label{rem_market-mptc_M}
Let us comment on the choice of $M = \R^d$ in this subsection.
After calculating the risk measure $\seq{\tilde{R}^{mar}}$, it is of course possible to choose a subspace of eligible portfolios $M$ and choose the capital requirements to be in that space (if $\tilde{R}_t^{mar}(X) \cap M_t$ is non-empty).

However, if the subspace $M \neq \R^d$ were to be chosen first and used for the recursive computation, then the last equality in \eqref{eqn_composed2} would no longer hold. This would cause several problems as the market extended multi-portfolio time consistent version, i.e. $\tilde{R}_t^{mar,M}(X) := \bigcup_{k \in \K_t} \bigcup_{Z \in \tilde{R}_{t+1}^{mar}(X-k)} R_t(-Z)\cap M_t$, while retaining the capital injection interpretation given above, is in general not multi-portfolio time consistent.  Furthermore, it would be more restrictive than the approach proposed above as  it holds $\tilde{R}_t^{mar,M}(X) \subseteq \tilde{R}_t^{mar}(X) \cap M_t$.

On the other hand, the multi-portfolio time consistent version of the market extension, while being multi-portfolio time consistent, does not admit a good economic interpretation for $M \neq \R^d$.

All of these problems disappear if $M = \R^d$. For this, and the motivation given at the beginning of this subsection, we suggest to use $M = \R^d$ when considering market extensions and if needed one can choose $u_t \in\tilde{R}_t^{mar}(X)\cap M_t$ afterwards.
\end{remark}

\begin{example}
\label{sec_shp}
We can use the algorithm from section~\ref{sec_alg} to compute the set of superhedging prices under either a conical or convex
market model. Dual representations for the set of superhedging portfolios are considered in e.g.
\cite{K99,S04,KS09,FR12} under proportional transaction costs; and in e.g. \cite{PP10,FR12b} under convex transaction costs.
We calculate the set of superhedging portfolios by computing the market-compatible
version of the worst case risk measure.  Let $\seq{R}$ be the worst case risk measure, that is $R_t:\LdzF{}\to\mathcal{P}(\LdzF{t};\LdzF{t,+})$ with \[R_t(X)=\{u\in \LdzF{t}: X+u\in \LdzF{+}\}.\]
The worst case risk measure $\seq{R}$ is conditionally convex and polyhedral with $R_{t,t+1}(-Z)[\omega_t]$ in \eqref{ui} given as the
upper image of a linear vector optimization problem.
By proposition~\ref{propoLVOP} one can calculate $\seq{R}$ $\omega_t$-wise since the worst case risk measure is multi-portfolio time
consistent.

Consider the market extension of the worst case risk measure, where trading is modeled by a sequence of solvency regions $\seq{K}$.
The multi-portfolio time consistent market extension $\seq{\tilde R^{mar}}$ with $\K_t = \LdpK{0}{t}{K_t}$ is nothing else than
the superhedging risk measure.  In particular, for a given claim $X \in \LdzF{}$, the set $SHP_t(X) := \tilde{R}_t^{mar}(-X)$ is the set of
superhedging portfolios of $X$.

Under proportional transaction costs, modeled by a sequence of solvency cones $\seq{K}$, by proposition~\ref{propoLVOP} and the discussion in section~\ref{sec_market}, the set of
superhedging portfolios $SHP_t(X)$ of $X$ can be calculated backwards in time by solving a sequence of linear vector
optimization problems \eqref{VOP} with ordering cone $K_t[\omega_t]$.  This backwards recursive algorithm is exactly the one proposed
in \cite{LR11}, see also \cite{RZ11}, which could be obtained as the simple structure of the worst case risk measure $\seq{R}$ yields a great simplification to the recursive structure \eqref{recursive}, respectively \eqref{eqn_composed_final},
\eqref{eqn_composed}.  Note that this simplification is specific to that example, which means that the method in
\cite{LR11,RZ11} cannot be generalized to other risk measures, whereas the approach discussed in this paper is widely applicable.
\end{example}

\section{Numerical examples}
\label{sec_ex}

In this section we will apply the algorithms for the recursive calculation of polyhedral and conditionally convex risk measures
presented in section~\ref{sec_alg}.  Specifically, we will consider the superhedging risk measure, relaxed worst case risk measure,
average value at risk, and entropic risk measure.

We consider a multi-dimensional tree that approximates the $d-1$ asset prices (denoted in the domestic currency) and assume stock price dynamics under the physical measure $\P$ are given by correlated geometric Brownian motions:
\[dS_t^i = S_t^i(\mu_i dt + \sigma_i dW_t^i), \quad i = 1,...,d-1\]
for Brownian motions $W^i$ and $W^j$ with correlation $\rho_{ij} \in [-1,1]$.  To create a tree for the correlated risky assets, we
follow the approach in~\cite{KM09}.  We expand the tree structure produced in such a method by allowing for $n^{d-1}$ (recombining)
branches for any natural number $n \geq 2$ and consider some maximum change from a parent to child node given by $\nu \in \R_{++}$
(instead of $2^{d-1}$ branches and $\nu = 1$ from the binomial model presented in~\cite{KM09}). 
That is, since every asset can rise or
fall, we consider the set of possible up-down scenarios given by
\[\mathcal{E} = \lrcurly{\transp{w_1,...,w_{d-1}}: w_i \in \lrcurly{-\nu, -\nu + \frac{2\nu}{n-1},..., \nu - \frac{2\nu}{n-1},\nu} \;
\forall i = 1,...,d-1}.\]
We note that in the situation with only a single risky asset, $n = 2$, and $\nu = 1$, this reduces to the Cox-Ross-Rubinstein binomial tree model.
To calculate the (conditional) probabilities of reaching a successor node, we partition the space $\R^{d-1}$ into $n^{d-1}$ boxes in such a way that each
element of $\mathcal{E}$ resides in a unique box, then the probability of rising or falling by level $e \in \mathcal{E}$ is given by
the probability of the (multivariate) normal distribution over the box containing $e$.

For simplicity, we additionally assume that the proportional transaction costs are constant for each of the risky assets, given by $\gamma = \transp{\gamma_1,...,\gamma_{d-1}} \in \R^{d-1}_+$ (possibly $0$).  Thus the bid and ask prices are given by $(S_t^b)^i = S_t^i(1-\gamma_i)$ and $(S_t^a)^i = S_t^i(1+\gamma_i)$ respectively for every $i = 1,...,d-1$.  In the case that $\gamma_i = 0$ then the bid-ask spread is $0$ for the $i^{th}$ risky asset.

Assume the existence of a risk-free asset with dynamics $\seq{B}$ and no bid-ask
spread, i.e. $B_t^b = B_t^a$ at all times $t$.  Further, we consider the case where cash (i.e. the risk-free asset) is an
intermediary for all transactions. That is, the exchange between any two assets is done via cash and not directly.
Under proportional transaction costs, the above simplifying assumptions ensure that the solvency cone $K_t$ at time $t$ is generated by the columns of the matrix
\[\lrparen{\begin{array}{cc}
      \transp{\frac{S_t^a}{B_t}} & -\transp{\frac{S_t^b}{B_t}} \\
      -I_{(d-1) \times (d-1)} & I_{(d-1) \times (d-1)}
    \end{array}}\]
where $I_{(d-1) \times (d-1)}$ denotes the identity matrix with $d-1$ rows and columns.

For the examples that contain convex transaction costs, we restrict ourselves to the two asset case (a single risky and a single risk-free asset; assets $1$ and $0$ respectively) and generate the solvency region $K_t$ at time $t$ by the dual equations: $k \in K_t$ if
\[K_t^+(k) := \lrparen{\begin{array}{c}
            k_0 + \theta_{t,0}\lrsquare{1-\exp\lrparen{-\frac{S_t^b k_1}{\theta_{t,0}}}} \\
            \theta_{t,1}\lrsquare{1-\exp\lrparen{-\frac{k_0}{S_t^a \theta_{t,1}}}} + k_1 \end{array}} \geq 0.
\]
In the above equations we specify a parameter $\theta_t \in \LdpK{\infty}{t}{\R^2_{++}}$ which defines the maximum number of risky asset that can be bought with the riskless asset $\theta_{t,1}$ and the maximum number of the riskless asset that can be bought with the risky asset $\theta_{t,0}$ at time $t$.  That is, beginning from the $0$ portfolio, it would not be possible to attain more than $\theta_{t,0}$ units of the risk-free asset or $\theta_{t,1}$ units of the risky asset.
 For simplicity, we assume that the trading strategy chosen does not impact the future market.  
 Thus, a trade at time $t$ does not affect the market at time $t+1$.

\begin{example}
\label{ex_avar_shp}
Consider a market with proportional transaction costs and two assets (a risk-free bond and a risky asset).  We choose as a market model
a recombining tree with 25 branches and $T = 9$ time steps over a one year time horizon.  Further consider the maximal
possible rise or fall in the Brownian motion to be given by $\nu = 2$.    Consider the market with high proportional transaction costs,
defined by $\gamma = 30\%$.

Let the risk-free rate of return be $10\%$.  Let the drift for the risky asset be $\mu = 12.5\%$ and the volatility given by $\sigma = 0.5$.  Consider the initial value of the risky asset to be $S_0 = \$100$ (measured in the risk-free asset).

Consider the superhedging risk measure and the average value at risk with constant parameter $\lambda = \transp{30\%,30\%}$ on the terminal payoff $X$ of an at the money European put
option, i.e. with strike price $\$100$.  Running the polyhedral algorithm presented in this paper,
the efficient frontier of the time $0$ superhedging risk measure and composed market extended average value at risk $\widetilde{AV@R}_0^{\lambda,mar}(X)$ is given by figure~\ref{fig_avar_shp}.  The circles in figure~\ref{fig_avar_shp} denote the vertices of the risk measures.  It is clear that the superhedging risk measure provides a more conservative set of risk compensating portfolios.  Note that the deviation from a line to the efficient frontier of the superhedging risk measure is due entirely to the transaction costs, which is in contrast to the average value of risk as the corners there are not solely determined by the transaction costs.
\begin{figure}[h]
\centering
\includegraphics[clip=true, trim=1.5in 3.4in 1.5in 3.5in, scale=.75]{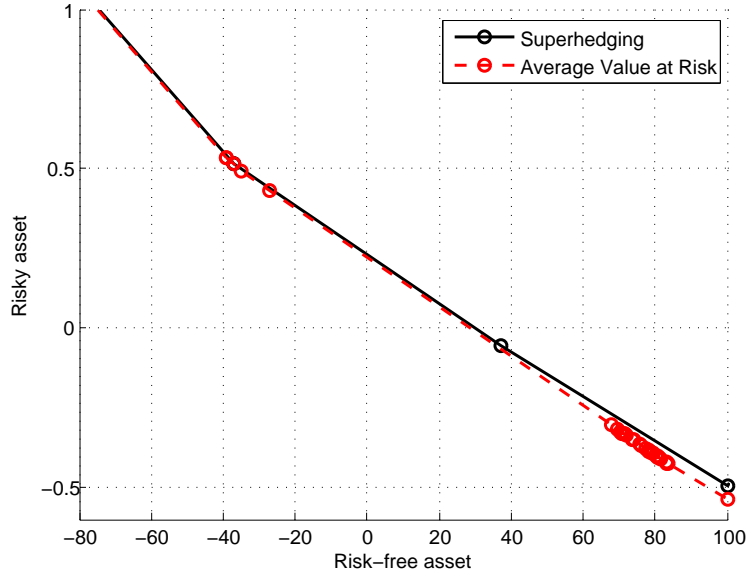}
\caption{Example~\ref{ex_avar_shp}: The superhedging risk measure and composed average value at risk under high proportional transaction costs}
\label{fig_avar_shp}
\end{figure}
\end{example}

\begin{example}
\label{ex_rwc}
Consider a market with proportional transaction costs and three assets (risk-free bond and two correlated risky assets).  We will
approximate the market with a binomial tree model with $T = 20$ time steps over a one year time horizon.  Consider a market with proportional
transaction costs defined by $\gamma = 5\%$.

Let the risk-free rate of return be $10\%$.  Let the drift for the risky assets be given by $\mu_1 = 15\%$ and $\mu_2 = 30\%$.  Let the volatility for the risky assets be given by $\sigma_1 = 0.5$ and $\sigma_2 = 1$.  Let the correlation be given by $\rho = 0.5$.  Consider the initial value of the risky assets to be $S_0 = \transp{\$1 , \$1}$ (measured in the risk-free asset).

Let $X$ be the terminal payoff of an outperformance option with strike price $K = \$1.10$, i.e. $X = \transp{-K I_{\{\max(S_T^a) \geq K\}} , I_{\{(S_T^a)^1 \geq (S_T^a)^2, (S_T^a)^1 \geq K\}} , I_{\{(S_T^a)^2 \geq (S_T^a)^1, (S_T^a)^2 \geq K\}}}$.

Consider the relaxed worst case risk measure with constant parameters $\varepsilon_i = .25$ for $i = 0,1,2$ and $G$ is the convex cone
generated by the vectors $\transp{1,-.25,-.25}$, $\transp{-.25,1,-.25}$, and $\transp{-.25,-.25,1}$.  The market extended
multi-portfolio time consistent version of the relaxed worst case risk measure can be calculated via the polyhedral algorithm
presented in this paper; a contour plot of the efficient frontier of the time $0$ risk measure $\tilde{R}_0^{mar}(X)$ is given by figure~\ref{fig_rwc_contour}.  Notice that, as the capital in the risk-less asset increases the level of capital necessary in the risky assets decreases.  However, increasing capital in one risky asset cannot totally offset a decrease in capital in the other risky asset, as evidenced by the curvature of the contour lines. 
\begin{figure}[h]
\centering
\includegraphics[clip=true, trim=1.5in 3.35in 1.5in 3.4in, scale=.75]{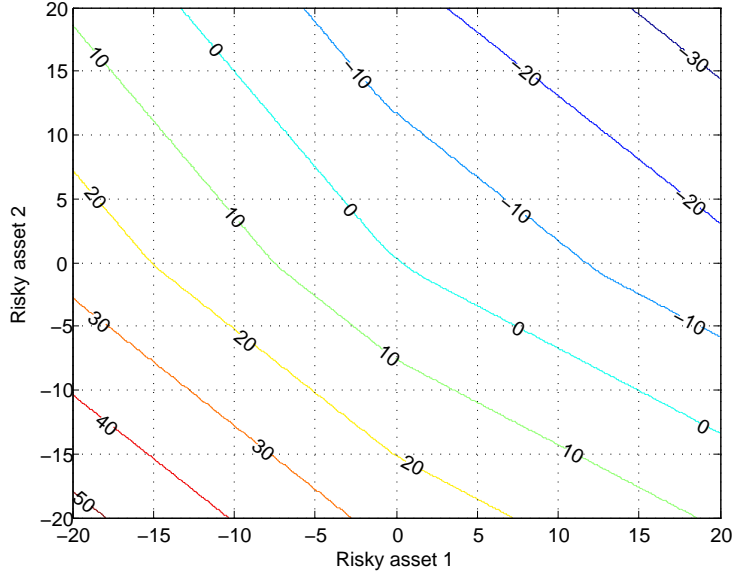}
\caption{Example~\ref{ex_rwc}: Contour plot of the efficient frontier for the composed relaxed worst case risk measure with 3 assets at differing levels of capital in the risk-less asset}
\label{fig_rwc_contour}
\end{figure}
\end{example}

\begin{example}
\label{ex_convex_rm}
Consider a market with convex transaction costs and two assets (risk-free bond and a risky asset).  We consider the 2 time step
Cox-Ross-Rubinstein model.  We consider a market with proportional transaction costs given by $\gamma = 5\%$ and convex transaction costs
given by $\theta_t = \transp{500,500}$ almost surely for each time point.

Let the risk-free rate of return be $0\%$.  Let the drift for the risky asset be $\mu = 12.5\%$ and the volatility given by $\sigma = 0.5$.  Consider the initial value of the risky asset to be $S_0 = \$1$ (measured in the risk-free asset).

Let $X$ be the terminal payoff from an out of the money binary option paying out $\$10$ with strike price $\$1.20$.  We are able to compute an approximation of the efficient frontier at time $0$ for the superhedging and entropic risk measures by running the convex algorithm presented in this text.
We consider two cases for the entropic risk measure, each with constant parameters $\lambda_i^t = 10\%$ for each asset and time and $C_t = \LdpK{0}{t}{C}$ where:
\begin{enumerate}
\item $C = \operatorname{cone}(\transp{1,0},\transp{0,1})$ be the convex cone generated by the vectors $\transp{1,0}$ and $\transp{0,1}$, i.e. the restrictive entropic risk measure; and
\item $C = \operatorname{cone}(\transp{1,-0.90},\transp{-0.90,1})$.
\end{enumerate}
Running the convex algorithm presented in this paper, with the approximation error at time $t = 0$ is given by $\epsilon < 0.30$ in all cases, the efficient frontier of the time $0$ composed market extended risk measures is shown in figures~\ref{fig_shp_entropic_big}, \ref{fig_shp_entropic_mid}, and \ref{fig_shp_entropic_small}.  The first plot, figure~\ref{fig_shp_entropic_big} shows that at a large enough scale the different risk measures all appear identical.  The curvature of the risk measures is also very evident at this size, the asymptotic behavior at $-1500$ in each asset is due to the choice of $\theta$.  In figure~\ref{fig_shp_entropic_mid}, we can see discrepancies between the superhedging risk measure and the two (composed) entropic risk measures.  However both (composed) entropic risk measures still appear to coincide.  In the final plot, figure~\ref{fig_shp_entropic_small}, the distinction between the least restrictive entropic risk measure is pronounced at this zoomed-in level of detail.  
The superhedging risk measure provides the most conservative estimate, in line with the theory.
\begin{figure}[H]
\centering
\includegraphics[clip=true, trim=1.5in 3.35in 1.5in 3.4in, scale=.75]{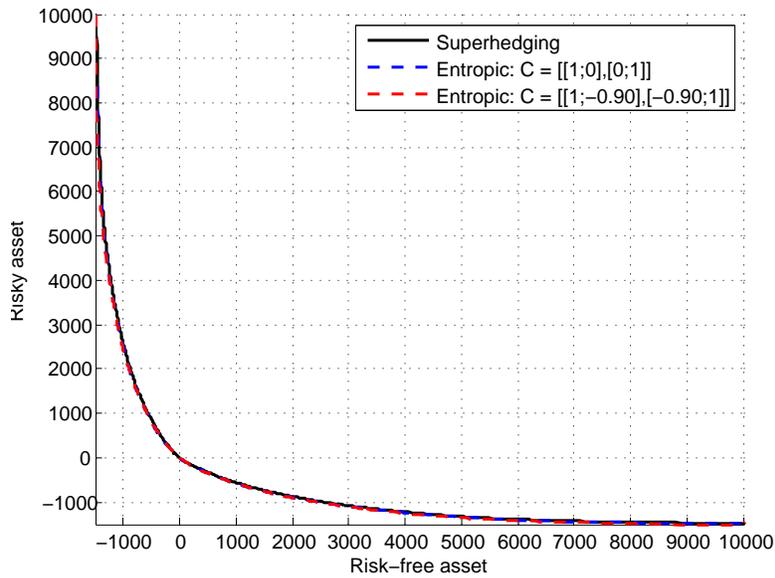}
\caption{Example~\ref{ex_convex_rm}: Convex superhedging and entropic risk measures under proportional and convex transaction costs, zoomed out view}
\label{fig_shp_entropic_big}
\end{figure}
\begin{figure}[H]
\centering
\includegraphics[clip=true, trim=1.5in 3.35in 1.5in 3.4in, scale=.75]{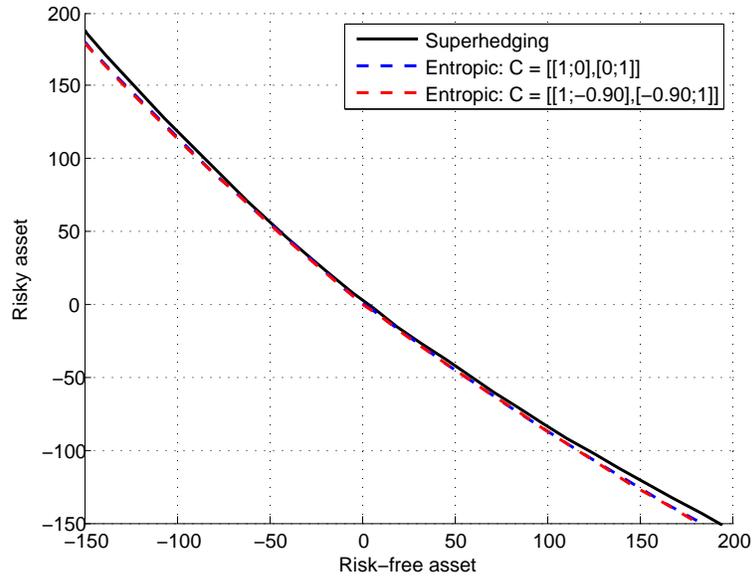}
\caption{Example~\ref{ex_convex_rm}: Convex superhedging price and entropic risk measures under proportional and convex transaction costs, mid-sized view}
\label{fig_shp_entropic_mid}
\end{figure}
\begin{figure}[H]
\centering
\includegraphics[clip=true, trim=1.5in 3.35in 1.5in 3.4in, scale=.75]{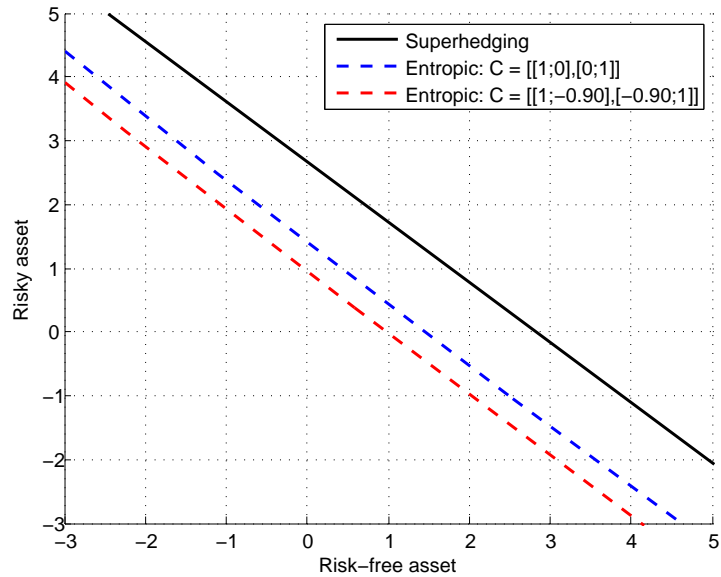}
\caption{Example~\ref{ex_convex_rm}: Convex superhedging price and entropic risk measures under proportional and convex transaction costs, near 0}
\label{fig_shp_entropic_small}
\end{figure}

\end{example}

\bibliographystyle{plain}
\bibliography{biblio}
\end{document}